\documentclass{PoS}
\usepackage{amsmath,amssymb,amsfonts,amsthm}

\newcommand{\pair}[2]{\left\langle #1,#2\right\rangle}

\def\vol{\mathrm{vol}}
\newtheorem{propo}{Proposition}
\newtheorem{theorem}{Theorem}
\newtheorem{lem}{Lemma}
\newtheorem{cor}{Corollary}

\theoremstyle{definition}
\newtheorem{defi}{Definition}

\theoremstyle{remark}
\newtheorem{remark}{Remark}

\title{Quantum Field Theory on Curved Noncommutative Spacetimes}

\ShortTitle{Quantum Field Theory on Curved Noncommutative Spacetimes}

\author{\speaker{Alexander Schenkel}%
       \\
        Institut f\"ur Theoretische Physik und Astrophysik\\
  Universit\"at W\"urzburg, Am Hubland, 97074 W\"urzburg, Germany\\
       E-mail: \email{aschenkel@physik.uni-wuerzburg.de}}

\abstract{
We summarize our recently proposed approach to quantum field theory on noncommutative
curved spacetimes. We make use of the Drinfel'd twist deformed differential geometry of Julius Wess and his group
in order to define an action functional for a real scalar field on a twist-deformed time-oriented, connected and 
globally hyperbolic Lorentzian manifold. The corresponding deformed wave operator admits unique deformed retarded and advanced 
Green's operators, provided we pose a support condition on the deformation. The solution space of the deformed wave equation is 
constructed explicitly and can be canonically equipped with a (weak) symplectic structure. 
The quantization of the solution space of the
deformed wave equation is performed using $\ast$-algebras over the ring $\mathbb{C}[[\lambda]]$. As a new result 
we add a proof that there exist symplectic isomorphisms between the deformed and the undeformed symplectic 
$\mathbb{R}[[\lambda]]$-modules. This immediately leads to $\ast$-algebra isomorphisms between the
deformed and the formal power series extension of the undeformed quantum field theory. 
The consequences of these isomorphisms are discussed.
}

\FullConference{Corfu Summer Institute on Elementary Particles and Physics - Workshop on Non Commutative Field Theory and Gravity,\\
		September 8-12, 2010\\
		Corfu Greece}

\begin{document}

\section{\label{sec:intro}Introduction}
Noncommutative (NC) geometry provides a rich mathematical framework to modify the standard formalism 
of quantum field theory (QFT)
in order to include quantum effects of spacetime itself. There are various proposals of how to combine QFT 
and NC geometry, see e.g.~the review articles \cite{NCQFTrev}, but most of these approaches are restricted to the
NC Euclidean space or to the NC Minkowski spacetime. However, to address questions in NC cosmology and NC black hole physics, where we
expect interesting physical effects to occur, it is essential to formulate QFT on NC curved spacetimes. 

Since QFT on commutative curved spacetimes is formulated in an elegant way using the algebraic approach, see e.g.~the monographs 
\cite{Wald,baer}, it is natural to proceed in this direction also in the NC case. One practical advantage of this formulation
is that it allows for a separate treatment of algebraic aspects of QFTs and issues concerning the choice of quantum state.
Recent approaches in this direction are due to Dappiaggi, Lechner and Morfa-Morales \cite{Dappiaggi:2010hc}
and myself and collaborators \cite{Ohl:2009qe,Schenkel:2010sc,Schenkel:2010jr}. In \cite{Dappiaggi:2010hc} Rieffel deformations
of local algebras of observables along Killing flows have been studied and applied to the construction of a deformed
Dirac field on curved spacetimes. In our approach \cite{Ohl:2009qe,Schenkel:2010sc,Schenkel:2010jr}
we have used methods from NC gravity \cite{Aschieri:2005yw, Aschieri:2005zs} to define an action functional
for a real scalar field on NC curved spacetimes and investigated the deformed wave equation and its solution space in detail
in terms of formal power series. We have shown by an explicit construction that the solution space can be equipped
 with a (weak) symplectic structure, which allows us to quantize these theories in terms of suitable $\ast$-algebras 
over the ring $\mathbb{C}[[\lambda]]$. Some first applications of our approach to QFT on 
NC curved spacetimes have been studied in \cite{Schenkel:2010sc,Schenkel:2010jr}.

The outline of the present proceedings article is as follows: In Section \ref{sec:tools} we collect the required tools
from Drinfel'd twists and their associated NC geometry, which are mainly taken from \cite{Aschieri:2005zs}. We define an action
functional for a real and free scalar field on twist-deformed curved spacetimes in Section \ref{sec:wave}
and derive the corresponding deformed wave operator. The deformed Green's operators are constructed in Section 
\ref{sec:greensoperators} and are applied in Section \ref{sec:symplvec} to construct the space of real solutions
of the deformed wave equation, which is then equipped with a (weak) symplectic structure. The quantization is
performed in Section \ref{sec:canquant}. As a new result we provide in Section \ref{sec:symplecto} 
a theorem that the deformed symplectic $\mathbb{R}[[\lambda]]$-module constructed in Section \ref{sec:symplvec} is isomorphic, 
via a symplectic isomorphism, to the formal power series extension of the undeformed symplectic vector space. 
The consequences for the deformed QFT are discussed in Section \ref{sec:consequences}. As another new element,
we develop throughout this paper a simplification of the formalism \cite{Ohl:2009qe}, where reality properties 
of the deformed QFT are more obvious and which is easier to apply. The equivalence to
 \cite{Ohl:2009qe} is shown in the Appendix \ref{app:equiv}. We conclude in Section \ref{sec:conc}.

\paragraph{Notation:} In this article we work in a formal deformation quantization setting. This means that we replace the
field $\mathbb{K}=\mathbb{R}$ or $\mathbb{C}$ by the commutative and unital ring $\mathbb{K}[[\lambda]]$,
 where $[[\lambda]]$ denotes formal power series in the deformation parameter $\lambda$.
Elements of $\mathbb{K}[[\lambda]]$ are given by $\mathbb{K}[[\lambda]] \ni \beta = \sum\limits_{n=0}^\infty \lambda^n \,\beta_{(n)}$,
where $\beta_{(n)}\in\mathbb{K}$ for all $n$.
 The sum and product on $\mathbb{K}[[\lambda]]$ reads, for all $\beta,\gamma\in\mathbb{K}[[\lambda]]$,
\begin{flalign}
\beta + \gamma := \sum_{n=0}^\infty\lambda^n (\beta_{(n)} + \gamma_{(n)})~,\quad \beta\,\gamma := 
\sum\limits_{n=0}^\infty \lambda^{n} \sum\limits_{m+k=n}\beta_{(m)}\gamma_{(k)}~.
\end{flalign} 

Let $V$ be a vector space over $\mathbb{K}$. Its formal power series
extension $V[[\lambda]]$ can be equipped with a $\mathbb{K}[[\lambda]]$-module structure by defining
\begin{flalign}
v+v^\prime:= \sum\limits_{n=0}^\infty \lambda^n (v_{(n)}+v^\prime_{(n)})~,\quad \beta\,v := \sum\limits_{n=0}^\infty\lambda^{n}
\sum\limits_{m+k=n} \beta_{(m)}\,v_{(k)}~,
\end{flalign}
for all $\beta\in\mathbb{K}[[\lambda]]$ and $v,v^\prime\in V[[\lambda]]$.
A $\mathbb{K}[[\lambda]]$-module homomorphism ($\mathbb{K}[[\lambda]]$-linear map) is a map between
two $\mathbb{K}[[\lambda]]$-modules preserving the $\mathbb{K}[[\lambda]]$-module structure. 
Note that a family of $\mathbb{K}$-linear maps  $P_{(n)}:V\to W$, $n\in \mathbb{N}^0$, induces a 
$\mathbb{K}[[\lambda]]$-linear map $P_\star:V[[\lambda]]\to W[[\lambda]]$ by defining for all $v\in V[[\lambda]]$
\begin{flalign}
\label{eqn:powerseriesmap}
 P_\star(v):= \sum\limits_{n=0}^\infty\lambda^n\,\sum\limits_{m+k=n} P_{(m)}(v_{(k)})~.
\end{flalign}
We shall use the notation $P_\star = \sum \lambda^n P_{(n)}$.
The other way around, let $P_\star:V[[\lambda]]\to W[[\lambda]]$  be a $\mathbb{K}[[\lambda]]$-linear map, then
it gives rise to a family of $\mathbb{K}$-linear maps $P_{(n)}:V\to W$, $n\in\mathbb{N}^0$, defined 
by $P_{(n)}(v):= \bigl(P_\star(v)\bigr)_{(n)}$, for all $v\in V$. $P_\star$ can be expressed in terms of the 
$P_{(n)}$ by (\ref{eqn:powerseriesmap})\footnote{
The $\mathbb{K}[[\lambda]]$-linear maps $P_\star$ and (\ref{eqn:powerseriesmap}) constructed 
by the $P_{(n)}$ are identical when acting on
polynomials $V[\lambda]$. Employing the $\lambda$-adic topology, one finds that 
$V[\lambda]$ is dense in $V[[\lambda]]$ and that all $\mathbb{K}[[\lambda]]$-linear maps $Q:V[[\lambda]]\to W[[\lambda]]$
are continuous. Thus, the maps are  identical on $V[[\lambda]]$.
}.

A (weak) symplectic $\mathbb{R}[[\lambda]]$-module $(W,\rho)$ is an $\mathbb{R}[[\lambda]]$-module
$W$ with an antisymmetric and $\mathbb{R}[[\lambda]]$-bilinear map $\rho:W\times W\to\mathbb{R}[[\lambda]]$, such that
$\rho(\varphi,\psi)=0$ for all $\psi\in W$ implies $\varphi =0$. Similar to \cite{baer} we suppress the term {\it weak}
in the following. An $\mathbb{R}[[\lambda]]$-linear map between symplectic $\mathbb{R}[[\lambda]]$-modules, 
which preserves the symplectic structure, will be simply called a symplectic map.

\section{\label{sec:tools}Basics on Drinfel'd twist deformed differential geometry}
We follow the approach of \cite{Aschieri:2005yw, Aschieri:2005zs} and refer to these works for details.
 Let $\mathcal{M}$ be a smooth manifold
and let $\Xi$ be the complexified vector fields on $\mathcal{M}$. A Drinfel'd twist is an invertible element 
$\mathcal{F}\in \bigl(U\Xi\otimes U\Xi\bigr)[[\lambda]]$, where $U\Xi$ is the universal enveloping algebra of $\Xi$,
satisfying
\begin{subequations}
\begin{flalign}
 &\mathcal{F}_{12}(\Delta\otimes \mathrm{id})\mathcal{F}=\mathcal{F}_{23}(\mathrm{id}\otimes\Delta)\mathcal{F}~,\\
 &(\epsilon\otimes \mathrm{id}) \mathcal{F} = 1 = (\mathrm{id}\otimes\epsilon)\mathcal{F}~,\\
 &\mathcal{F}=1\otimes 1 +\mathcal{O}(\lambda)~,
\end{flalign}
\end{subequations}
where $\mathcal{F}_{12}:= \mathcal{F}\otimes 1 $ and $\mathcal{F}_{23}:=1\otimes \mathcal{F}$. The map $\Delta:U\Xi[[\lambda]]\to
\bigl(U\Xi\otimes U\Xi\bigr)[[\lambda]]$ is the canonical coproduct
 and $\epsilon:U\Xi[[\lambda]] \to \mathbb{C}[[\lambda]]$ is the canonical counit. We denote the inverse twist by
 (sum over $\alpha$ understood)
\begin{flalign}
\mathcal{F}^{-1}=\bar f^\alpha\otimes \bar f_\alpha \in \bigl(U\Xi \otimes U\Xi\bigr)[[\lambda]]~.
\end{flalign}

To simplify our investigations we demand the twist to be real, i.e.~$\mathcal{F}^{\ast\otimes\ast}=(S\otimes S)\mathcal{F}_{21}$,
 and to satisfy $S(\bar f^\alpha)\,\bar f_\alpha=1$, where $S:U\Xi[[\lambda]]
\to U\Xi[[\lambda]]$ is the canonical antipode.
Note that the so-called abelian twists \cite{Reshetikhin:1990ep,Jambor:2004kc}
\begin{flalign}
\mathcal{F}_\text{abelian}=\exp\left(-\frac{i\lambda}{2}\Theta^{ab}X_a\otimes X_b\right)~,
\end{flalign}
where $X_a\in\Xi$ are mutually commuting real vector fields and 
$\Theta^{ab}$ is real, constant and antisymmetric, are part of the Drinfel'd twists we consider.

The commutative algebra of smooth complex valued functions $\mathcal{A}=\bigl(C^\infty(\mathcal{M}),\cdot\bigr)$
 is deformed into a NC associative algebra  $\mathcal{A}_\star:= \bigl(C^\infty(\mathcal{M})[[\lambda]],\star\bigr)$ 
 by introducing the $\star$-product
 \begin{flalign}
 h\star k := \bar f^\alpha( h )\cdot \bar f_\alpha( k )~,
 \end{flalign}
 for all $h,k\in C^\infty(\mathcal{M})[[\lambda]]$.
 The vector fields act on functions via the Lie derivative.
 Due to the reality property of the twist the $\star$-product is hermitian, i.e.~$(h\star k)^\ast= k^\ast\star h^\ast$, for all
 $h,k\in \mathcal{A}_\star$.

Similarly, we deform the differential calculus $(\Omega^\bullet,\wedge,d)$ of differential forms on $\mathcal{M}$
  into a deformed differential calculus $\Omega^\bullet_\star:=(\Omega^\bullet[[\lambda]],\wedge_\star,d)$. The deformed wedge product is defined by
\begin{flalign}
\omega\wedge_\star \omega^\prime := \bar f^\alpha( \omega )\wedge  \bar f_\alpha ( \omega^\prime )~,
\end{flalign}
for all $\omega, \omega^\prime\in \Omega^\bullet[[\lambda]]$. The vector fields act via the Lie derivative on differential forms.
It turns out that the undeformed exterior differential $d$ satisfies the Leibniz rule
\begin{flalign}
d(\omega\wedge_\star \omega^\prime) = (d\omega)\wedge_\star \omega^\prime + (-1)^{\text{deg}(\omega)} \omega\wedge_\star(d\omega^\prime)~,
\end{flalign}
for all $\omega,\omega^\prime\in\Omega^\bullet[[\lambda]]$, since Lie derivatives commute with $d$. 
Note that the space of formal power series of $n$-forms $\Omega_\star^n:=\Omega^{n}[[\lambda]]$ is an $\mathcal{A}_\star$-bimodule,
where the left and right $\mathcal{A}_\star$-action is provided by the deformed  wedge product.

We deform the vector fields $\Xi$ into an $\mathcal{A}_\star$-bimodule of deformed vector fields $\Xi_\star := \Xi[[\lambda]]$ by 
employing the deformed left and right $\mathcal{A}_\star$-action
\begin{flalign}
h\star v:=\bar f^\alpha(h)\,\bar f_\alpha(v)~,\quad v\star h:=\bar f^\alpha(v)\,\bar f_\alpha(h)~,
\end{flalign}
for all $h\in\mathcal{A}_\star$ and $v\in\Xi_\star$. 
The action of the twist  on $\Xi[[\lambda]]$ is given by the Lie derivative.
The duality pairing between vector fields and one-forms can also be deformed by
\begin{flalign}
\pair{v}{\omega}_\star:= \pair{\bar f^\alpha(v)}{\bar f_\alpha(\omega)}~,
\end{flalign}
for all $v\in\Xi_\star$ and $\omega\in\Omega^1_\star$. The $\star$-pairing 
with $v$ on the right and $\omega$ on the left is defined analogously. One obtains the following relevant property
by using identities of the twist
\begin{flalign}
\label{eqn:pairinglinearity}
\pair{h\star v\star k}{\omega\star l}_\star=h\star\pair{v}{k \star\omega}_\star \star l~,
\end{flalign}
for all $h,k,l\in\mathcal{A}_\star$, $v\in\Xi_\star$ and $\omega\in \Omega^1_\star$.

Employing the $\star$-tensor product
\begin{flalign}
\tau\otimes_\star\tau^\prime := \bar f^\alpha(\tau)\otimes\bar f_\alpha(\tau^\prime)~,
\end{flalign}
we can deform the tensor algebra $(\mathcal{T},\otimes)$ over $\mathcal{A}$
generated by $\Xi$ and $\Omega^1$ into the tensor algebra $(\mathcal{T}_\star,\otimes_\star)$ over $\mathcal{A}_\star$ 
generated by $\Xi_\star$ and $\Omega^1_\star$.
The relation (\ref{eqn:pairinglinearity}) extends to the $\star$-tensor algebra 
\begin{flalign}
\pair{\tau\otimes_\star v\star h}{\omega\otimes_\star\tau^\prime}_\star= \tau\otimes_\star \pair{v}{h\star \omega}_\star \star\tau^\prime~,
\end{flalign}
for all $h\in\mathcal{A}_\star$, $v\in\Xi_\star$, $\omega\in \Omega^1_\star$ and $\tau,\tau^\prime\in\mathcal{T}_\star$.

The integral over top-forms $\omega\in\Omega_\star^{\text{dim}(\mathcal{M})}$ is defined by 
$\int_\mathcal{M}\omega := \sum\lambda^n\int_\mathcal{M}\omega_{(n)}$.
Due to the assumption $S(\bar f^\alpha)\,\bar f_\alpha=1$ the integral satisfies the graded cyclicity property
\cite{Aschieri:2009ky}
\begin{flalign}
\label{eqn:gradcyc}
\int_\mathcal{M} \omega\wedge_\star \omega^\prime = (-1)^{\text{deg}(\omega) \text{deg}(\omega^\prime)}\int_\mathcal{M} \omega^\prime \wedge_\star \omega = \int_\mathcal{M} \omega\wedge \omega^\prime~,
\end{flalign}
for all $\omega,\omega^\prime\in \Omega_\star^\bullet$ with 
$\text{deg}(\omega)+\text{deg}(\omega^\prime)=\text{dim}(\mathcal{M})$ and $\text{supp}(\omega)\cap\text{supp}(\omega^\prime)$
 compact\footnote{%
Let $\omega := \sum \lambda^n \omega_{(n)}\in\Omega^\bullet[[\lambda]]$ and $\omega^\prime := \sum \lambda^n \omega^\prime_{(n)}\in\Omega^\bullet[[\lambda]]$.
The statement $\text{supp}(\omega)\cap\text{supp}(\omega^\prime)$ compact is an abbreviation for
 $\text{supp}(\omega_{(n)})\cap\text{supp}(\omega^\prime_{(m)})$ compact for all $n,m\in\mathbb{N}^0$.
}.

\section{\label{sec:wave}Deformed action functional and wave operator for a free and real scalar field}
Using the tools of Section \ref{sec:tools} we are in the position to construct a 
deformed action functional for a free and real  scalar field $\Phi$.
For this we additionally require a metric field and a volume form. 
Consider a classical Lorentzian manifold $(\mathcal{M},g)$ with metric field $g\in \Omega^1\otimes\Omega^1$ and corresponding
volume form $\vol_g\in \Omega^{\text{dim}(\mathcal{M})}$. To be as general as possible we 
consider as the deformed metric field and volume form elements $g_\star\in \Omega_\star^1\otimes_\star\Omega_\star^1$
and $\vol_\star\in\Omega_\star^{\text{dim}(\mathcal{M})}$ only subject to the conditions
$g_\star^{\ast}=g_\star$, $\vol_\star^\ast=\vol_\star$, $g_\star\vert_{\lambda=0}=g$ and $\vol_\star\vert_{\lambda=0}=\vol_g$.
The $\star$-inverse deformed metric field $g^{-1}_\star\in \Xi_\star\otimes_\star\Xi_\star$, defined
by $\langle g^{-1}_\star,\langle g_\star,v\rangle_\star\rangle_\star=v$ and 
$\langle g_\star,\langle g^{-1}_\star,\omega\rangle_\star\rangle_\star=\omega$ for all $v\in\Xi_\star$ and
 $\omega\in\Omega_\star^1$,
 exists, is hermitian and satisfies $g^{-1}_\star\vert_{\lambda=0}=g^{-1}$.
We follow \cite{Ohl:2009qe} and define
\begin{flalign}
\label{eqn:action}
 S_\star[\Phi] :=-\frac{1}{2}\int_\mathcal{M}\Bigl(\pair{\pair{d\Phi}{g^{-1}_\star}_\star}{d\Phi}_\star + M^2 \,\Phi\star\Phi\Bigr)\star \text{vol}_\star~.
\end{flalign}
Varying this action by functions $\delta\Phi$ of compact support we obtain the top-form valued wave operator 
\begin{flalign}
\label{eqn:topeom}
\tilde P_\star(\varphi) := \frac{1}{2}\Bigl(\square_\star(\varphi)\star \vol_\star + \vol_\star\star \bigl(\square_\star(\varphi^\ast) \bigr)^\ast - M^2 \varphi\star \vol_\star - M^2\vol_\star\star\varphi\Bigr)~.
\end{flalign}
Employing the $\mathbb{C}[[\lambda]]$-module isomorphism 
$\star_g: C^\infty(\mathcal{M})[[\lambda]]\to\Omega^{\text{dim}(\mathcal{M})}[[\lambda]]\,,~h\mapsto h\,\vol_g$, 
given by the {\it undeformed} Hodge operator corresponding to $g$ we define a scalar valued operator 
$P_\star: C^\infty(\mathcal{M})[[\lambda]]\to C^\infty(\mathcal{M})[[\lambda]]$ by $P_\star:=\star_g^{-1}\circ \tilde P_\star$.
The wave operator $P_\star$ is formally self-adjoint with respect to the {\it undeformed} scalar product
\begin{flalign}
\label{eqn:scalarproduct}
\bigl(\varphi,\psi\bigr) := \int_\mathcal{M} \varphi^\ast\,\psi\,\vol_g~.
\end{flalign}
More precisely, 
\begin{flalign}
\label{eqn:selfadjoint}
\bigl(\varphi,P_\star(\psi)\bigr) = \bigl(P_\star(\varphi),\psi\bigr)~
\end{flalign}
holds true for all 
$\varphi,\psi \in C^\infty(\mathcal{M})[[\lambda]]$ with $\text{supp}(\psi)\cap\text{supp}(\varphi) $
 compact. 
\begin{remark}
\label{rem:conventions}
The definition of the scalar valued wave operator is different to \cite{Ohl:2009qe}. In \cite{Ohl:2009qe} we have used
the {\it deformed} Hodge operator $\star_\star:C^\infty(\mathcal{M})[[\lambda]]\to\Omega^{\text{dim}(\mathcal{M})}[[\lambda]]\,,~
h\mapsto h\star\vol_\star$, to extract a scalar valued operator from $\tilde P_\star$ (\ref{eqn:topeom}). 
The resulting operator was formally self-adjoint with respect to the deformed scalar product 
$(\varphi,\psi)_\star=\int_\mathcal{M}\varphi^\ast\star\psi\star\vol_\star$. Since the map
$\iota:=\star_g^{-1} \circ \star_\star:C^\infty(\mathcal{M})[[\lambda]]\to C^\infty(\mathcal{M})[[\lambda]]$ 
is a $\mathbb{C}[[\lambda]]$-module isomorphism, the theory we are going to construct in the following is completely equivalent 
to the one in \cite{Ohl:2009qe}, but has the advantage that reality properties of the deformed QFT are 
more obvious in the present formulation. We show the equivalence of both formulations in the Appendix \ref{app:equiv}.
\end{remark}

\section{\label{sec:greensoperators}The deformed Green's operators}
In this section we assume $(\mathcal{M},g)$ to be a time-oriented, connected and globally hyperbolic Lorentzian manifold.
Guided by Section \ref{sec:wave} we consider  $\mathbb{C}[[\lambda]]$-linear maps
 $P_\star:C^\infty(\mathcal{M})[[\lambda]]\to C^\infty(\mathcal{M})[[\lambda]]$,
 which are deformations of normally hyperbolic operators, i.e.~$P_\star\vert_{\lambda=0}$ is normally hyperbolic.
 This allows us to treat also more general wave operators than the deformed Klein-Gordon operators of Section \ref{sec:wave}.
 We can write $P_\star$ as a formal power series of operators
 \begin{flalign}
 P_\star := \sum\limits_{n=0}^\infty \lambda^n P_{(n)}~.
\end{flalign}
 From the considerations in Section \ref{sec:wave} we find that it is natural to demand $P_\star$ to be formally self-adjoint
 with respect to the undeformed scalar product (\ref{eqn:scalarproduct}) and to be real,
  i.e.~$\left(P_\star(\varphi)\right)^\ast=P_\star(\varphi^\ast)$ for all $\varphi\in C^\infty(\mathcal{M})[[\lambda]]$.
  Furthermore, we demand $P_{(n)}$ to be finite-order differential operators for all $n\in \mathbb{N}$.

An interesting question is the existence and uniqueness of deformed  Green's operators
$\Delta_{\star\pm}:C^\infty_0(\mathcal{M})[[\lambda]]\to C^\infty(\mathcal{M})[[\lambda]]$ corresponding to $P_\star$.
These operators will play an important role in the construction of the deformed QFT. 
Based on the strong results for the $\lambda^0$-part \cite{baer}, we have shown in \cite{Ohl:2009qe}
that the deformed Green's operators exist, provided we assume a support condition on the NC corrections $P_{(n)}$, $n>0$.
\begin{theorem}[\cite{Ohl:2009qe}]
\label{theo:green}
 Let $(\mathcal{M},g)$ be a  time-oriented, connected and globally hyperbolic Lorentzian
manifold. Let $P_\star:=\sum\lambda^n P_{(n)}$ be a formal deformation of a normally 
hyperbolic operator acting on $C^\infty(\mathcal{M})[[\lambda]]$, where for $n>0$ the
$P_{(n)}:C^\infty(\mathcal{M})\to C^\infty_0(\mathcal{M})$  are finite-order differential operators.
Then there exist unique deformed Green's operators $\Delta_{\star\pm}:=\sum \lambda^n\Delta_{(n)\pm}$ satisfying
\begin{subequations}
\label{eqn:green}
\begin{flalign}
\label{eqn:green1}
& P_\star\circ\Delta_{\star\pm} = \text{id}_{C_0^\infty(\mathcal{M})}~,\\\label{eqn:green2}
& \Delta_{\star\pm} \circ P_\star\big\vert_{C_0^\infty(\mathcal{M})[[\lambda]]} = \text{id}_{C_0^\infty(\mathcal{M})}~,\\\label{eqn:green3}
& \text{supp}(\Delta_{(n)\pm}(\varphi))\subseteq J_{\pm}(\text{supp}(\varphi))~,\quad \text{for all~}n\in\mathbb{N}^0~\text{and~}\varphi\in C_0^\infty(\mathcal{M})~,
\end{flalign}
\end{subequations}
where $J_\pm$ is the causal future/past with respect to the classical metric $g$.\\
The explicit expressions for $\Delta_{(n)\pm}$, $n>0$, read
\begin{flalign}
\label{eqn:explicitgreen}
\Delta_{(n)\pm} = \sum\limits_{k=1}^{n}\sum\limits_{j_1=1}^n \dots\sum\limits_{j_k=1}^n(-1)^k \delta_{j_1+\dots+j_k, n} ~
  \Delta_{\pm}\circ P_{(j_1)}\circ \Delta_\pm \circ P_{(j_2)} \circ \dots \circ \Delta_\pm \circ P_{(j_k)}\circ \Delta_{\pm}~,
\end{flalign}
where $\delta_{n,m}$ is the Kronecker-delta.
\end{theorem}
The support condition $P_{(n)}:C^\infty(\mathcal{M})\to C^\infty_0(\mathcal{M})$ is a remnant of the formal power series
construction and it is sufficient to make all compositions in (\ref{eqn:explicitgreen}) well-defined. Similar issues occur in
perturbative QFT, where one has to assume the interaction to be of compact support. For the deformed Klein-Gordon
operator of Section \ref{sec:wave} the support condition is 
satisfied for the following two scenarios: 1.) we deform by twists of compact support 2.) the vector fields entering the twist
are asymptotic (outside a region of compact support) Killing vector fields of $g_\star$. In other cases we require, similar to
perturbative QFT, an infrared regularization in terms of compactly supported deformations.

Next, we study properties of the deformed retarded-advanced Green's operator given by the $\mathbb{C}[[\lambda]]$-linear map
\begin{flalign}
\Delta_{\star} := \Delta_{\star +} - \Delta_{\star -}:C_0^\infty(\mathcal{M})[[\lambda]] \to C_\mathrm{sc}^\infty(\mathcal{M})[[\lambda]]~,
\end{flalign}
where $C_\mathrm{sc}^\infty(\mathcal{M})$ are the functions of spatially compact support. 
The importance of this map lies in the fact that it defines the covariant Poisson 
bracket relations (i.\,e.~the Peierls bracket relations) of classical field theory and
 the canonical commutation relations of QFT.
 We have obtained the following
\begin{theorem}[\cite{Ohl:2009qe}]
\label{theo:complex}
 Let $(\mathcal{M},g)$ be a  time-oriented, connected and globally hyperbolic Lorentzian manifold 
and let $P_\star$ and $\Delta_{\star\pm}$ be as above. Then the sequence of $\mathbb{C}[[\lambda]]$-linear maps
\begin{flalign}
\label{eqn:complex}
 0\longrightarrow C_0^\infty(\mathcal{M})[[\lambda]] \stackrel{P_\star}{\longrightarrow} C_0^\infty(\mathcal{M})[[\lambda]] \stackrel{\Delta_\star}{\longrightarrow} C_\mathrm{sc}^\infty(\mathcal{M})[[\lambda]] \stackrel{P_\star}{\longrightarrow}C_\mathrm{sc}^\infty(\mathcal{M})[[\lambda]]
\end{flalign}
is a complex, which is exact everywhere.
\end{theorem}
This theorem provides us information on the solution space of the deformed wave equation, i.e.~the functions
$\Phi\in C^\infty_\mathrm{sc}(\mathcal{M})[[\lambda]]$ satisfying $P_\star(\Phi)=0$. Combining this information,
 we find that the factor  space $C^\infty_0(\mathcal{M})[[\lambda]]/P_\star[C^\infty_0(\mathcal{M})[[\lambda]]]$ is isomorphic,
  as a $\mathbb{C}[[\lambda]]$-module, to the space of complex solutions of the deformed wave equation $\text{Sol}_{P_\star}^\mathbb{C}$.
 The isomorphism is given by the map 
 \begin{flalign}
 \mathcal{I}_\star^{\mathbb{C}}: C^\infty_0(\mathcal{M})[[\lambda]]/P_\star[C^\infty_0(\mathcal{M})[[\lambda]]]\to \text{Sol}_{P_\star}^\mathbb{C}~,\quad [\varphi]\mapsto\Delta_\star(\varphi)~,
\end{flalign}
which is well-defined since $\Delta_\star\circ P_\star =0$ when acting on $C^\infty_0(\mathcal{M})[[\lambda]]$.

\section{\label{sec:symplvec}The space of real solutions of the deformed wave equation}
Since our purpose is to describe a real scalar field we have to restrict $\text{Sol}_{P_\star}^\mathbb{C}$ to the space of real
solutions, denoted by $\text{Sol}_{P_\star}^\mathbb{R}$. Due to our new definition of
the deformed scalar valued wave operator, see Remark \ref{rem:conventions}, this task turns out to be simpler
 than the construction presented in \cite{Ohl:2009qe}. Since the wave operator $P_\star$ and the Green's operators $\Delta_{\star\pm}$
 are real, we find the isomorphism
 \begin{flalign}
 \label{eqn:realiso}
 \mathcal{I}_\star: C^\infty_0(\mathcal{M},\mathbb{R})[[\lambda]]/P_\star[C^\infty_0(\mathcal{M},\mathbb{R})[[\lambda]]]\to \text{Sol}_{P_\star}^\mathbb{R}~,\quad [\varphi]\mapsto\Delta_\star(\varphi)~.
 \end{flalign}
 We define $V_\star:=C^\infty_0(\mathcal{M},\mathbb{R})[[\lambda]]/P_\star[C^\infty_0(\mathcal{M},\mathbb{R})[[\lambda]]]$.
 
The next step is to show that $V_\star$ carries a symplectic structure. We define the $\mathbb{R}[[\lambda]]$-bilinear map
\begin{flalign}
\label{eqn:symplectic}
\omega_\star: V_\star\times V_\star\to \mathbb{R}[[\lambda]]~,\quad ([\varphi],[\psi])\mapsto
\omega_\star([\varphi],\,[\psi])=\bigl(\varphi,\Delta_\star(\psi)\bigr)~,
\end{flalign} 
where we have employed the scalar product (\ref{eqn:scalarproduct}).
Analogously to \cite{Ohl:2009qe} one shows that $\omega_\star$ is weakly non-degenerate and antisymmetric.

We summarize the result of this section in the following
\begin{propo}
Let $(\mathcal{M},g)$ be a  time-oriented, connected and globally hyperbolic Lorentzian manifold 
and let $P_\star$ and $\Delta_{\star\pm}$ be as above. Then there is a canonically associated symplectic $\mathbb{R}[[\lambda]]$-module
$(V_\star,\omega_\star)$, where $V_\star =C^\infty_0(\mathcal{M},\mathbb{R})[[\lambda]]/P_\star[C^\infty_0(\mathcal{M},\mathbb{R})[[\lambda]]]$ and $\omega_\star$ is given in (\ref{eqn:symplectic}).
\end{propo}
Note that since $V_\star$ is isomorphic to $\text{Sol}_{P_\star}^\mathbb{R}$ via (\ref{eqn:realiso}) this means
that the space of real solutions of the deformed wave equation can be equipped with a symplectic structure.

\section{\label{sec:canquant}Canonical quantization}
Guided by the algebraic approach to commutative QFTs, see e.g.~\cite{baer}, we quantize the symplectic
$\mathbb{R}[[\lambda]]$-module $(V_\star,\omega_\star)$ using suitable $\ast$-algebras of field observables.
Since our present focus is on
formal deformation quantization, the preferred choice of an algebra of observables is the $\ast$-algebra
of field polynomials and not the Weyl algebra. 
This is due to the fact that formal power series prohibit the use of $C^\ast$-algebras.
For a review on algebras, states and representations in deformation quantization see \cite{waldmann}.
We make the following
\begin{defi}
\label{def:fieldpoly}
 Let $(W,\rho)$ be a symplectic $\mathbb{R}[[\lambda]]$-module.
 A unital $\ast$-algebra over $\mathbb{C}[[\lambda]]$ is called $\ast$-algebra of 
 field polynomials $A_{(W,\rho)}$, if it is generated by the
 elements $\Phi( \varphi )$, $\varphi \in W$, subject to the relations
\begin{subequations}
\label{eqn:fieldalgebra}
\begin{flalign}
~~\label{eqn:fieldoplin}\Phi(\beta\,\varphi+\gamma\,\psi)&=\beta\,\Phi(\varphi)+\gamma\,\Phi(\psi)~,\\
~~\Phi(\varphi)^\ast&=\Phi(\varphi)~,\\
~~[\Phi(\varphi),\Phi(\psi)]&= i\,\rho(\varphi,\psi)\, 1~,
\end{flalign}
\end{subequations}
for all $\varphi,\psi\in W$ and $\beta,\gamma\in \mathbb{R}[[\lambda]]$.
\end{defi}
\noindent A state on a unital $\ast$-algebra $A$ over $\mathbb{C}[[\lambda]]$ 
is a $\mathbb{C}[[\lambda]]$-linear map $\Omega:A\to\mathbb{C}[[\lambda]]$ satisfying
\begin{subequations}
\begin{flalign}
 \Omega(1)&=1~,\\
 \Omega(a^\ast\,a)&\geq 0~,~\forall a\in A~.
\end{flalign}
\end{subequations}
The ordering on $\mathbb{R}[[\lambda]]$ is defined by
\begin{flalign}
 \mathbb{R}[[\lambda]]\ni\gamma= \sum\limits_{n=n_0}^{\infty}\lambda^n \gamma_{(n)}>0\quad :\Longleftrightarrow\quad \gamma_{(n_0)}>0~.
\end{flalign}

The observables for the deformed QFT are given by the algebra $A_{(V_\star,\omega_\star)}$. Due to the relations
(\ref{eqn:fieldalgebra}) it is natural to interpret $\Phi(\varphi)$ as smeared field operators 
``~$\Phi(\varphi)= \int_\mathcal{M}\Phi\,\varphi\,\vol_g$~''.
Fixing a state $\Omega_\star$ on $A_{(V_\star,\omega_\star)}$ we can use the formal GNS construction outlined in
\cite{waldmann} to represent $A_{(V_\star,\omega_\star)}$ on a pre-Hilbert space over $\mathbb{C}[[\lambda]]$.
The choice of state $\Omega_\star$ for the deformed QFT is, similar to the commutative case, in general highly nonunique.
However, as we will show in Section \ref{sec:consequences}, there is a way to induce states on the algebra $A_{(V_\star,\omega_\star)}$
by pulling-back states $\Omega$ of the commutative QFT.

\section{\label{sec:symplecto}Symplectic isomorphisms}
In this section we provide a theorem showing that the deformed symplectic $\mathbb{R}[[\lambda]]$-module 
$(V_\star,\omega_\star)$ is isomorphic,
via a symplectic isomorphism, to the formal power series extension of the undeformed symplectic vector space $(V[[\lambda]],\omega)$,
where $V:=C^\infty_0(\mathcal{M},\mathbb{R})/P[C^\infty_0(\mathcal{M},\mathbb{R})]$ and
\begin{flalign}
 \omega:V[[\lambda]]\times V[[\lambda]] \to\mathbb{R}[[\lambda]]\,,~([\varphi],[\psi])\mapsto \omega([\varphi],[\psi]) 
= \bigl(\varphi,\Delta(\psi)\bigr)~.
\end{flalign} 
Here $P=P_\star\vert_{\lambda=0}$ 
is the undeformed wave operator and  $\Delta=\Delta_\star\vert_{\lambda=0}$ is the corresponding undeformed 
retarded-advanced Green's operator. Analogously to (\ref{eqn:realiso}) we have the isomorphism
\begin{flalign}
 \label{eqn:realisoclass}
\mathcal{I}:V[[\lambda]] \to \text{Sol}_P^\mathbb{R}\,,~[\varphi]\mapsto \Delta(\varphi)~.
\end{flalign}

\begin{lem}
\label{lem:solutionspaceiso}
 The maps $T_\pm:\text{Sol}_{P_\star}^\mathbb{R}\to \text{Sol}_{P}^\mathbb{R}$ defined by 
\begin{flalign}
\label{eqn:soliso}
T_\pm:=\text{id}_{C^\infty_\text{sc}(\mathcal{M})}+t_{\pm}=
\text{id}_{C^\infty_\text{sc}(\mathcal{M})}+\Delta_{\pm}\circ\sum\limits_{n=1}^\infty\lambda^n\,P_{(n)}
\end{flalign}
are $\mathbb{R}[[\lambda]]$-module isomorphisms.
\end{lem}
\begin{proof}
The compositions $\Delta_\pm\circ P_{(n)}$ are well-defined due to the support condition on
$P_{(n)}$, $n>0$. Additionally, the maps $\Delta_\pm\circ P_{(n)}$ are real for all $n>0$.
 The composition of $P$ and $T_\pm$ is given by
\begin{flalign}
 P\circ T_\pm = P + P\circ \Delta_\pm\circ\sum\limits_{n=1}^\infty \lambda^n\, P_{(n)} =
P + \sum\limits_{n=1}^\infty \lambda^n\,P_{(n)} =\sum\limits_{n=0}^\infty \lambda^n\,P_{(n)}=P_\star~,
\end{flalign}
where we have used that $P\circ\Delta_\pm = \text{id}_{C^\infty_0(\mathcal{M})}$.
Thus, for all $\Phi\in\text{Sol}_{P_\star}^\mathbb{R}$ we have $T_\pm(\Phi)\in\text{Sol}_P^\mathbb{R}$.\newline
The inverse of $T_\pm$ is constructed by the geometric series
\begin{flalign}
 T_\pm^{-1}=\left(\text{id}_{C^\infty_\text{sc}(\mathcal{M})} +t_\pm\right)^{-1}
=\sum\limits_{m=0}^\infty \left(-t_\pm\right)^m \stackrel{\text{(\ref{eqn:explicitgreen})}}{=} 
\text{id}_{C^\infty_\text{sc}(\mathcal{M})} -\Delta_{\star\pm}\circ\sum\limits_{n=1}^\infty\lambda^n\,P_{(n)}
\end{flalign}
and satisfies
\begin{flalign}
 P_\star\circ T_\pm^{-1} &= P_\star - P_\star\circ\Delta_{\star\pm}\circ \sum\limits_{n=1}^\infty\lambda^n\,P_{(n)}
= P_\star - \sum\limits_{n=1}^\infty\lambda^n~P_{(n)} = P~,
\end{flalign}
where we have used that $P_\star\circ\Delta_{\star\pm}=\text{id}_{C^\infty_0(\mathcal{M})}$.
Thus, for all $\Phi\in\text{Sol}_{P}^\mathbb{R}$ we have $T^{-1}_\pm(\Phi)\in\text{Sol}_{P_\star}^\mathbb{R}$.
\end{proof}
It turns out that in general $T_+$ and $T_-$ differ. To see this let $\Phi\in\text{Sol}_{P_\star}^\mathbb{R}$ be arbitrary.
Due to Theorem \ref{theo:complex} there is a $\varphi\in C^\infty_0(\mathcal{M},\mathbb{R})[[\lambda]]$, 
such that $\Phi=\Delta_\star(\varphi)$.
We obtain 
\begin{flalign}
T_+(\Phi) - T_-(\Phi) = \Delta\left(\sum\limits_{n=1}^\infty\lambda^n\,P_{(n)}\left(\Delta_\star(\varphi)\right)\right)
=-\Delta\left(P\left(\Delta_\star(\varphi)\right)\right)~.
\end{flalign}
The difference between $T_+$ and $T_-$ is thus given by the operator $\Delta\circ P\circ \Delta_\star$. 
Notice that this operator is {\it not} zero in general, since the relation $\Delta\circ P =0$ just holds when acting on functions
of compact support, while $\Delta_\star$ maps to functions of noncompact support. 
To be more explicit we expand the operator $\Delta\circ P\circ\Delta_\star$ to first order in $\lambda$
 by using Theorem \ref{theo:green} and find
\begin{flalign}
\nonumber \Delta\circ P\circ \Delta_\star &= \Delta\circ P\circ \bigl(\Delta - \lambda(\Delta_+\circ P_{(1)}\circ\Delta_+ - \Delta_-\circ P_{(1)}\circ \Delta_-) \bigr)+\mathcal{O}(\lambda^2) \\
&= -\lambda \,\Delta\circ P_{(1)}\circ \Delta +\mathcal{O}(\lambda^2)~.
\end{flalign} 
Since $P_{(1)}$ comes from the choice of deformation, while $\Delta$ describes the commutative dynamics, 
these operators are independent and $\Delta\circ P\circ \Delta_\star$ in general does not vanish.
\begin{remark}
The maps $T_\pm$ can be interpreted as retarded/advanced isomorphisms, since they 
 depend on the retarded/advanced Green's operators.
Due to the support property of $\Delta_\pm$ and the support condition
on $P_{(n)}$, $n>0$, we obtain order by order in $\lambda$ that $T_\pm(\Phi)$ is equal to $\Phi$ for sufficiently small/large times,
i.e.~for $t\to \mp\infty$.
\end{remark}

Employing the isomorphisms $T_\pm$, the isomorphism $\mathcal{I}_\star$ (\ref{eqn:realiso}) 
and its commutative counterpart $\mathcal{I}$ (\ref{eqn:realisoclass}), we obtain the $\mathbb{R}[[\lambda]]$-module isomorphisms
$\mathbf{T}_\pm :=\mathcal{I}^{-1}\circ T_\pm\circ\mathcal{I}_\star: V_\star\to V[[\lambda]]$. We can
map the deformed symplectic $\mathbb{R}[[\lambda]]$-module $(V_\star,\omega_\star)$, via a symplectic isomorphism, 
to the symplectic $\mathbb{R}[[\lambda]]$-module $(V[[\lambda]],\widehat{\omega}_\star)$, where by definition
\begin{flalign}
\widehat{\omega}_\star\bigl([\varphi],[\psi]\bigr) := \omega_\star\bigl(\mathbf{T}^{-1}_\pm\bigl([\varphi]\bigr),\mathbf{T}^{-1}_\pm\bigl([\psi]\bigr) \bigr)~,
\end{flalign}
for all $[\varphi],[\psi]\in V[[\lambda]]$. This expression can be simplified and we obtain 
\begin{flalign}
\nonumber \widehat{\omega}_\star\bigl([\varphi],[\psi]\bigr) 
&= \Bigl(\mathbf{T}^{-1}_\pm([\varphi]),\Delta_\star\bigl(\mathbf{T}^{-1}_\pm([\psi])\bigr)\Bigr)
 = \Bigl(\mathbf{T}^{-1}_\pm([\varphi]),T^{-1}_\pm\bigl(\Delta(\psi)\bigr)\Bigr)\\
&=\Bigl(T_\pm^{-1\dagger}\bigl(\mathbf{T}^{-1}_\pm([\varphi])\bigr),\Delta(\psi)\Bigr)
=-\Bigl(\Delta\bigl(T_\pm^{-1\dagger}\bigl(\mathbf{T}^{-1}_\pm([\varphi])\bigr)\bigr),\psi\Bigr)~,
\end{flalign}
where we have used the adjoint map
\begin{flalign}
T_\pm^{-1\dagger} = \text{id}_{C^\infty_0(\mathcal{M})} - \sum\limits_{n=1}^\infty\lambda^n\,P_{(n)}\circ \Delta_{\star\mp}
\end{flalign}
and that $\Delta$ is antihermitian. Defining the map $\widehat{\Delta}_\star:C^\infty_0(\mathcal{M},\mathbb{R})[[\lambda]]\to 
\text{Sol}_{P}^\mathbb{R}$ by
\begin{flalign}
\widehat{\Delta}_\star := \Delta\circ T_\pm^{-1\dagger} \circ \mathcal{I}_\star^{-1}\circ T^{-1}_\pm\circ \Delta
\end{flalign}
we have for all $[\varphi],[\psi]\in V[[\lambda]]$
\begin{flalign}
\widehat{\omega}_\star\bigl([\varphi],[\psi]\bigr) = \bigl(\varphi,\widehat{\Delta}_\star(\psi)\bigr)~.
\end{flalign}
Analogously to (\ref{eqn:realiso}) and (\ref{eqn:realisoclass}) we define the isomorphism 
\begin{flalign}
 \widehat{\mathcal{I}}_\star: V[[\lambda]] \to \text{Sol}_{P}^\mathbb{R}~,\quad [\varphi]\mapsto\widehat{\Delta}_\star(\varphi)~.
\end{flalign}
We obtain the following
\begin{theorem}
 \label{theo:symplecto}
The map $\mathbf{S}: V[[\lambda]]\to V[[\lambda]]$ defined by $\mathbf{S} = \sum\limits_{n=0}^\infty\lambda^n S_{(n)}$,
where
\begin{subequations}
\label{eqn:}
 \begin{flalign}
S_{(0)} &= \text{id}_V~,\\
S_{(1)} &= \frac{1}{2} \mathcal{I}^{-1}\circ\widehat{\mathcal{I}}_{(1)}~,\\
 S_{(n)} &= \frac{1}{2}\Bigl(\mathcal{I}^{-1}\circ \widehat{\mathcal{I}}_{(n)} - \sum\limits_{m=1}^{n-1}S_{(m)}\circ S_{(n-m)}\Bigr) ~,~\forall n\geq 2~,
 \end{flalign}
\end{subequations}
provides a symplectic isomorphism between $(V[[\lambda]],\widehat{\omega}_\star)$ and $(V[[\lambda]],\omega)$,
i.e.~for all $[\varphi],[\psi]\in V[[\lambda]]$ we have $\widehat{\omega}_\star([\varphi],[\psi])=\omega(\mathbf{S}[\varphi],\mathbf{S}[\psi])$.
\end{theorem}
\begin{proof}
 The map $\mathbf{S}$ is invertible, since it is a formal deformation of the identity map.
Let $[\varphi],[\psi]\in V[[\lambda]]$ be arbitrary. We obtain
\begin{flalign}
 \nonumber \omega(\mathbf{S}[\varphi],\mathbf{S}[\psi]) &= \sum\limits_{n=0}^\infty\lambda^n\sum\limits_{m+k+i+j=n}
\omega(S_{(m)}[\varphi]_{(k)},S_{(i)}[\psi]_{(j)})\\
\label{eqn:zwischenschritt}&= \sum\limits_{n=0}^\infty\lambda^n\sum\limits_{k+j+l=n}~\sum\limits_{m+i=l}
\omega(S_{(m)}[\varphi]_{(k)},S_{(i)}[\psi]_{(j)})~.
\end{flalign}
Note that for all $[\varphi],[\psi]\in V$ we have $\omega([\varphi],S_{(0)}[\psi])=\omega(S_{(0)}[\varphi],[\psi])$ (trivially) and 
\begin{flalign}
 \omega([\varphi],S_{(1)}[\psi])=\frac{1}{2}\widehat{\omega}_{(1)}([\varphi],[\psi])=-\frac{1}{2}\widehat{\omega}_{(1)}([\psi],[\varphi])
= \omega(S_{(1)}[\varphi],[\psi])~.
\end{flalign}
By induction it follows that $\omega([\varphi],S_{(n)}[\psi])=\omega(S_{(n)}[\varphi],[\psi])$ for all $n\geq 0$.

Using this, the inner sum of (\ref{eqn:zwischenschritt}) reads
\begin{flalign}
 \sum\limits_{m+i=l}
\omega(S_{(m)}[\varphi]_{(k)},S_{(i)}[\psi]_{(j)}) 
= \omega([\varphi]_{(k)},\sum\limits_{m=0}^l S_{(m)}\circ S_{(l-m)}[\psi]_{(j)})~.
\end{flalign}
It remains to simplify the map $\sum_{m=0}^l S_{(m)}\circ S_{(l-m)}$. For $l=0$ this is simply the identity map
and for $l=1$ it reads $2\, S_{(1)}=\mathcal{I}^{-1}\circ \widehat{\mathcal{I}}_{(1)}$. For $l\geq 2$ we find
\begin{flalign}
 \sum\limits_{m=0}^l S_{(m)}\circ S_{(l-m)} = 2\,S_{(l)} + \sum\limits_{m=1}^{l-1}S_{(m)}\circ S_{(l-m)}
 = \mathcal{I}^{-1}\circ\widehat{\mathcal{I}}_{(l)}~.
\end{flalign}
Thus, (\ref{eqn:zwischenschritt}) reads
\begin{flalign}
 \omega(\mathbf{S}[\varphi],\mathbf{S}[\psi]) =  
\sum\limits_{n=0}^\infty\lambda^n\sum\limits_{k+j+l=n}\omega([\varphi]_{(k)},\mathcal{I}^{-1}\circ \widehat{\mathcal{I}}_{(l)}[\psi]_{(j)})
= \widehat{\omega}_{\star}([\varphi],[\psi])~.
\end{flalign}

\end{proof}
As a direct consequence we obtain
\begin{cor}
\label{cor:symplecto}
 The maps $\mathbf{S}_\pm:=\mathbf{S}\circ \mathbf{T}_\pm:V_\star\to V[[\lambda]]$ are symplectic isomorphisms between
 $(V_\star,\omega_\star)$ and $(V[[\lambda]],\omega)$, i.e.~$\omega\bigl(\mathbf{S}_\pm([\varphi]),\mathbf{S}_\pm([\psi])\bigr)
=\omega_\star\bigl([\varphi],[\psi]\bigr)$ for all $[\varphi],[\psi]\in V_\star$.
\end{cor}

\section{\label{sec:consequences}Consequences of the symplectic isomorphisms}
In this section we study consequences of the symplectic isomorphisms $\mathbf{S}_\pm$ of Corollary \ref{cor:symplecto}.
\subsection{$\ast$-algebra of field polynomials}
The symplectic isomorphisms $\mathbf{S}_\pm:V_\star\to V[[\lambda]]$ canonically induce $\ast$-algebra isomorphisms
$\mathfrak{S}_\pm:A_{(V_\star,\omega_\star)}\to A_{(V[[\lambda]],\omega)}$ between the $\ast$-algebra of field polynomials
of the deformed and the formal power series extension of the undeformed QFT. The maps $\mathfrak{S}_\pm$
are defined on the generators $1,\Phi_\star([\varphi])\in A_{(V_\star,\omega_\star)}$ by
\begin{subequations}
\begin{flalign}
 \mathfrak{S}_\pm\bigl(1\bigr) &= 1~,\\
 \mathfrak{S}_\pm\bigl(\Phi_\star([\varphi])\bigr)&=\Phi\bigl(\mathbf{S}_\pm([\varphi])\bigr) ~,
\end{flalign}
\end{subequations}
and extended to $ A_{(V_\star,\omega_\star)}$ as $\ast$-algebra homomorphisms. Here $\Phi([\psi])$, $[\psi]\in V[[\lambda]]$, 
are the generators of the algebra $A_{(V[[\lambda]],\omega)}$.
In short, we obtain
\begin{propo}
\label{propo:algebraiso}
 There exist $\ast$-algebra isomorphisms $\mathfrak{S}_\pm:A_{(V_\star,\omega_\star)}\to A_{(V[[\lambda]],\omega)}$.
\end{propo}

This means that we can {\it mathematically} describe the NC QFT in terms of a 
formal power series extension of the corresponding commutative QFT.
 However, the {\it physical interpretation} has to be adapted properly: If we want to probe the 
NC QFT with a set of smearing functions $\lbrace [\varphi_i]\rbrace$ in order to extract
physical observables (e.g.~Wightman functions) we have to probe the commutative QFT with a different set of smearing 
functions $\lbrace \mathbf{S}_\pm([\varphi_i]) \rbrace$.

In \cite{Schenkel:2010jr} we have shown that $\ast$-algebra isomorphisms similar to Proposition \ref{propo:algebraiso}
also exist for a class of convergent deformations.
The difference there is that the deformed QFT is isomorphic to a reduced undeformed QFT, where certain strongly localized
observables are excluded.

\subsection{Symplectic automorphisms} 
An important class of symmetries of QFTs are those which are induced by symplectic automorphisms.
In particular, isometries of the background fall into this class.
\begin{defi}
\label{def:symplecticauto}
 Let $(W,\rho)$ be a symplectic $\mathbb{R}[[\lambda]]$-module. A map $\alpha\in\text{End}_{\mathbb{R}[[\lambda]]}(W)$ 
is called a symplectic automorphism, if it is invertible and if $\rho\bigl(\alpha\varphi,\alpha\psi\bigr)=\rho(\varphi,\psi)$, 
for all $\varphi,\psi\in W$.
The set $\mathcal{G}_{(W,\rho)}$ of all symplectic automorphisms with the usual composition of homomorphisms $\circ$ forms 
the group of symplectic automorphisms.
\end{defi}
Due to the symplectic isomorphisms $\mathbf{S}_\pm:V_\star\to V[[\lambda]]$ we find
\begin{propo}
\label{propo:symplecticgroup}
There exist group isomorphisms $\mathbf{S}_{\pm\mathcal{G}}:\mathcal{G}_{(V[[\lambda]],\omega)}\to\mathcal{G}_{(V_\star,\omega_\star)}
,~\alpha\mapsto \mathbf{S}_\pm^{-1}\circ \alpha\circ \mathbf{S}_\pm$.
\end{propo}
The symplectic automorphisms $\mathcal{G}_{(W,\rho)}$ of a symplectic $\mathbb{R}[[\lambda]]$-module $(W,\rho)$ canonically induce
$\ast$-algebra isomorphisms of the corresponding $\ast$-algebra of field polynomials $A_{(W,\rho)}$.
Thus, the deformed QFT enjoys the same amount of  symmetries as the undeformed one.
However, the transformations are represented in a non-canonical, and in general also non-geometric,
way by using the symplectic isomorphisms $\mathbf{S}_\pm$.

In \cite{Schenkel:2010jr} we have shown that Proposition \ref{propo:symplecticgroup} is restricted to the formal deformation
quantization setting. Convergent deformations in general break some of the symplectic automorphisms.

\subsection{Algebraic states}
Due to the symplectic isomorphisms $\mathbf{S}_\pm$ the space of algebraic states on $A_{(V_\star,\omega_\star)}$ and $A_{(V[[\lambda]],\omega)}$ can be related in a precise way. To explain this we require the following well-known
\begin{lem}
\label{lem:pullback}
 Let $A_1$ and $A_2$ be two unital $\ast$-algebras over $\mathbb{C}[[\lambda]]$ and let 
$\kappa:A_1 \to A_2$ be a $\ast$-algebra homomorphism. Then each state $\Omega_2$ on $A_2$
induces a state $\Omega_1$ on $A_1$ by defining
\begin{flalign}
 \Omega_1(a):= \Omega_2\bigl(\kappa (a)\bigr)~,
\end{flalign}
for all $a\in A_1$.
\end{lem}
The proof of this standard lemma can be found e.g.~in \cite{Schenkel:2010jr}. 
The  state $\Omega_1$ is called the pull-back of $\Omega_2$.
Employing the $\ast$-algebra isomorphisms $\mathfrak{S}_\pm$ of Proposition \ref{propo:algebraiso}
we obtain the following
\begin{propo}
\label{propo:states}
The $\ast$-algebra isomorphisms $\mathfrak{S}_\pm$ provide bijections between the states on $A_{(V[[\lambda]],\omega)}$ 
and the states on $A_{(V_\star,\omega_\star)}$. 
The $\mathcal{G}_{(V[[\lambda]],\omega)}$-symmetric states on $A_{(V[[\lambda]],\omega)}$ are pulled-back 
to $\mathcal{G}_{(V_\star,\omega_\star)}$-symmetric states on $A_{(V_\star,\omega_\star)}$, and vice versa.
\end{propo} 
The proof of this statement can be found in \cite{Schenkel:2010jr}.

\section{\label{sec:conc}Conclusions and outlook}
In this proceedings article we have summarized our recently developed approach to QFT on NC curved spacetimes. We have constructed 
a deformed action functional for a real and free scalar field on NC curved spacetimes
 by employing methods from twist-deformed differential geometry. The deformed wave operator and its corresponding
 Green's operators were constructed explicitly in terms of formal power series. The solution space of the deformed wave equation 
 was constructed explicitly,  equipped with a symplectic structure and quantized in terms of $\ast$-algebras over the ring 
 $\mathbb{C}[[\lambda]]$. We have shown that the deformed symplectic $\mathbb{R}[[\lambda]]$-module is isomorphic, 
 via symplectic isomorphisms,
 to the formal power series extension of the undeformed symplectic vector space. A direct consequence of this symplectic isomorphism
 for the deformed QFT is that it is $\ast$-algebra isomorphic to the formal power series extension of the undeformed QFT.
 This immediately yields isomorphisms between the corresponding groups of symplectic automorphisms and bijections between
 the corresponding spaces of algebraic states.
 
 In future work it would be interesting to study more examples of convergent deformations and their properties, as it was already
 initiated in \cite{Schenkel:2010jr}. Furthermore, a more detailed study of QFTs on deformed
  cosmological and black hole spacetimes, see e.g.~\cite{Schenkel:2010sc}, and their phenomenology would be interesting.

\appendix
\section*{Acknowledgements}
I want to thank the organizers and participants of the  \textit{Corfu Summer Institute on Elementary Particles and Physics 2010} for 
this very interesting conference. Additionally I want to thank my friends, colleagues and collaborators
for the discussions and comments on this work.
This research is supported by Deutsche Forschungsgemeinschaft through the Research 
Training Group GRK\,1147 \textit{Theoretical Astrophysics and Particle Physics}.

\section{\label{app:equiv}Equivalence to \cite{Ohl:2009qe}}
We show that the formalism presented in this paper is equivalent to \cite{Ohl:2009qe}
by providing the corresponding isomorphisms. 
Quantities in the formulation [5] are distinguished from those of the present paper by a bar. 
\paragraph{Wave operators:}
As already stated in Remark \ref{rem:conventions} in Section \ref{sec:wave}, the scalar valued wave operator
$P_\star$ is related to the wave operator $\bar{P}_\star$ by the $\mathbb{C}[[\lambda]]$-module isomorphism
\begin{flalign}
\label{eqn:hodgeiso}
\iota:C^\infty(\mathcal{M})[[\lambda]]\to C^\infty(\mathcal{M})[[\lambda]]\,,~\varphi\mapsto 
\star_g^{-1}\bigl(\varphi\star \vol_\star\bigr)~,
\end{flalign}
where $\star_g$ is the undeformed Hodge operator corresponding to $g$. More precisely, we have
\begin{flalign}
 P_\star = \iota\circ \bar{P}_\star~.
\end{flalign}
The deformed and undeformed scalar products are related by
\begin{subequations}
\label{eqn:scalarproductrelation}
\begin{flalign}
 &\bigl(\varphi,\iota(\psi)\bigr)=\int_\mathcal{M}\varphi^\ast\, \iota(\psi)\,\vol_g = 
\int_\mathcal{M}\varphi^\ast\,\left( \psi\star\vol_\star\right)
\stackrel{\text{(\ref{eqn:gradcyc})}}{=}\int_\mathcal{M}\varphi^\ast\star \psi\star\vol_\star=\bigl(\varphi,\psi\bigr)_\star~,\\
&\bigl(\iota(\varphi),\psi\bigr)=\int_\mathcal{M}\left(\iota(\varphi)\,\vol_g\right)^\ast\, \psi = 
\int_\mathcal{M}\left(\varphi\star \vol_\star\right)^\ast\,\psi
\stackrel{\text{(\ref{eqn:gradcyc})}}{=}\int_\mathcal{M}\varphi^\ast\star \psi\star\vol_\star=\bigl(\varphi,\psi\bigr)_\star~,
\end{flalign}
\end{subequations}
for all $\varphi,\psi\in C^\infty(\mathcal{M})[[\lambda]]$. Let $\bar{P}_\star$ be formally self-adjoint with respect to
$(~,~)_\star$, then $P_\star = \iota\circ \bar{P}_\star$ is formally self-adjoint with respect to $(~,~)$,
since
\begin{flalign}
\bigl(\varphi,P_\star(\psi)\bigr)= \bigl(\varphi,\bar{P}_\star(\psi)\bigr)_\star = \bigl(\bar{P}_\star(\varphi),\psi\bigr)_\star =
\bigl(P_\star(\varphi),\psi\bigr)~.
\end{flalign}
The reverse direction is shown analogously.
\paragraph{Green's operators:}
Let $\bar{\Delta}_{\star\pm}$ be the Green's operators corresponding to $\bar{P}_\star$. Then
the Green's operators corresponding to $P_\star=\iota\circ\bar{P}_\star$ are given by
\begin{flalign}
 \Delta_{\star\pm} = \bar{\Delta}_{\star\pm}\circ \iota^{-1}~.
\end{flalign}
Thus, the retarded-advanced Green's operator is given by
\begin{flalign}
 \Delta_\star = \bar{\Delta}_\star\circ \iota^{-1}~.
\end{flalign}
\paragraph{Symplectic \text{$\mathbb{R}[[\lambda]]$}-modules:}
In \cite{Ohl:2009qe} we have defined the $\mathbb{R}[[\lambda]]$-module
\begin{flalign}
 H :=\lbrace \varphi\in C^\infty_0(\mathcal{M})[[\lambda]]: \bigl(\bar{\Delta}_{\star\pm}(\varphi)\bigr)^\ast 
=\bar{\Delta}_{\star\pm}(\varphi) \rbrace~,
\end{flalign}
which was used as a pre-symplectic $\mathbb{R}[[\lambda]]$-module.
This space is isomorphic to $C^\infty_0(\mathcal{M},\mathbb{R})[[\lambda]]$ via the isomorphism $\iota^{-1}$.
To see this let $\varphi\in C^\infty_0(\mathcal{M},\mathbb{R})[[\lambda]]$ be arbitrary, then $\iota^{-1}(\varphi)\in H$ since
\begin{flalign}
 \bigl(\bar{\Delta}_{\star\pm}\left(\iota^{-1}(\varphi)\right)\bigr)^\ast = \bigl(\Delta_{\star\pm}(\varphi)\bigr)^\ast = 
\Delta_{\star\pm}(\varphi) = \bar{\Delta}_{\star\pm}\left(\iota^{-1}(\varphi)\right)~.
\end{flalign}
Let now $\varphi\in H$ be arbitrary, then there is a $\psi\in C^\infty_0(\mathcal{M})[[\lambda]]$, such that
$\varphi = \iota^{-1}(\psi)$. We find that $\psi$ is real, since
\begin{flalign}
 0=\bigl(\bar{\Delta}_{\star\pm}(\varphi)\bigr)^\ast-\bar{\Delta}_{\star\pm}(\varphi) = \Delta_{\star\pm}\left(\psi^\ast-\psi\right)
\quad\Rightarrow\quad \psi^\ast=\psi~.
\end{flalign}

The symplectic $\mathbb{R}[[\lambda]]$-module in \cite{Ohl:2009qe} was defined by 
$\bar{V}_\star:= H/\bar{P}_\star[C^\infty_0(\mathcal{M},\mathbb{R})[[\lambda]]]$, while 
$V_\star = C^\infty_0(\mathcal{M},\mathbb{R})[[\lambda]]/ P_\star[C^\infty_0(\mathcal{M},\mathbb{R})[[\lambda]]]$. The isomorphism
$\iota: H\to C^\infty_0(\mathcal{M},\mathbb{R})[[\lambda]]$ gives rise to an isomorphism between the factor spaces, since
\begin{flalign}
 \iota\bigl(\varphi + \bar{P}_\star(\psi)\bigr) = \iota(\varphi) + P_\star(\psi)~,
\end{flalign}
for all $\varphi\in H$ and $\psi\in C^\infty_0(\mathcal{M},\mathbb{R})[[\lambda]]$.
The symplectic structure on $\bar{V}_\star$ and $V_\star$ is given by
\begin{subequations}
\begin{flalign}
 \bar{\omega}_\star([\varphi],[\psi]) &= \bigl(\varphi, \bar{\Delta}_\star(\psi)\bigr)_\star~,\\
 \omega_\star([\varphi],[\psi]) &= \bigl(\varphi,\Delta_\star(\psi)\bigr)~.
\end{flalign}
\end{subequations}
We obtain that the map $\iota:\bar{V}_\star\to V_\star$ is a symplectic isomorphism
\begin{flalign}
 \omega_\star([\iota(\varphi)],[\iota(\psi)]) = \bigl(\iota(\varphi),\Delta_\star(\iota(\psi))\bigr) 
\stackrel{\text{(\ref{eqn:scalarproductrelation})}}{=} \bigl(\varphi,\bar{\Delta}_\star(\psi)\bigr)_\star
 = \bar{\omega}_\star([\varphi],[\psi])~.
\end{flalign}
\paragraph{$\ast$-algebras of field polynomials:} The symplectic isomorphism $\iota$ immediately leads
to a $\ast$-algebra isomorphism between $A_{(\bar{V}_\star,\bar{\omega}_\star)}$ and $A_{(V_\star,\omega_\star)}$.
Thus, the approach presented in this paper leads to a QFT which is mathematically equivalent to the one
obtained in \cite{Ohl:2009qe}.
However, the physical interpretation has to be adapted properly: The operators $\Phi_\star(\varphi)\in A_{(V_\star,\omega_\star)}$
should be interpreted as smeared field operators with respect to the smearing 
``~$\Phi_\star(\varphi)= \int_\mathcal{M} \Phi\,\varphi\,\vol_g$~'', while the operators 
$\bar\Phi_\star(\varphi)\in A_{(\bar{V}_\star,\bar{\omega}_\star)}$
should be interpreted as smeared with the $\star$-products 
``~$\bar{\Phi}_\star(\varphi)=\int_\mathcal{M}\Phi\star\varphi\star\vol_\star$~''.

\end{document}